\documentclass[a4paper]{article}

\usepackage{RR}

\usepackage{hyperref}

\usepackage{amssymb}
\usepackage{graphicx}
\usepackage{a4wide}
\usepackage{latexsym}

\usepackage{amsmath}
\usepackage{amsfonts}
\usepackage{epsfig}
\usepackage{amsmath, amsthm, amssymb}
\usepackage{hyperref}
\usepackage{verbatim}
\usepackage[english]{babel}
\usepackage{mathrsfs}

\newtheorem{Theorem}{Theorem}
\newtheorem{Proposition}[Theorem]{Proposition}

\newtheorem{Lemma}[Theorem]{Lemma}

\newtheorem{Remark}[Theorem]{Remark}

\RRNo{6475}

\RRdate{June 2007}%% date de publication du rapport

\RRauthor{% les auteurs  Premier auteur, avec une note
Natalia Osipova \thanks{INRIA Sophia Antipolis, France, e-mail:
Natalia.Osipova@sophia.inria.fr}  \thanks{The work was supported
by France Telecom R\&D Grant ``Mod\'elisation et Gestion du Trafic
R\'eseaux Internet'' no. 42937433.}}

\authorhead{N. Osipova}

\RRtitle{Comparaison des politiques DPS}

\RRetitle{{Comparison of the Discriminatory Processor Sharing
Policies }}
\titlehead{{Comparison of the Discriminatory Processor Sharing
Policies }}

\RRresume{ L'ordre de service DPS (Discriminatory Processor
Sharing) qui était introduit par Kleinrock est un problème très
intéressant et peut être appliqué dans beaucoup de domaines comme
les télécommunications, les applications web et la modélisation de
flux TCP. Avec le DPS, les jobs qui viennent dans le système sont
contrôlés par un vecteur de poids. En modifiant le vecteur de
poids, il est possible de contrôler les taux de service des jobs,
donner la priorité à certaines classes de jobs et optimiser
certaines caractéristiques du système. Le problème du choix des
poids est donc très important et très difficile en raison de la
complexité du système. Dans le présent papier, nous comparons deux
politiques DPS avec les vecteurs de poids différents et nous
présentons des résultats sur la monotonicité du temps moyen de
service du système en fonction du vecteur de poids, sous certaines
conditions sur le système. Le système devrait consister en
plusieurs classes avec des moyennes très différentes. Pour les
classes qui ont une moyenne très proche il faut choisir les meme poids.
}

\RRabstract{Discriminatory Processor Sharing policy introduced by
Kleinrock is of  a great interest in many application areas,
including telecommunications, web applications and TCP flow
modelling. Under the DPS policy the job priority is controlled by
a vector of weights. Varying the vector of weights, it is possible
to modify the service rates of the jobs and optimize system
characteristics. In the present paper we present results
concerning the comparison of two DPS policies with different
weight vectors. We show the monotonicity of the expected sojourn
time of the system depending on the weight vector under certain
condition on the system. Namely, the system has to consist of
classes with means which are quite different from each other. For
the classes with similar means we suggest to select the same
weights.} %%

\RRmotcle{Discriminatory Processor Sharing, le temp de service
exponentielle, optimisation.}

\RRkeyword{Discriminatory Processor Sharing, exponential service
times, optimization.}

\RRprojet{MAESTRO}
\RRtheme{\THCom}
\URSophia %% pour ceux qui sont au Sud.

\begin{document}

\makeRR

\section{Introduction}

The Discriminatory Processor Sharing (DPS) policy was introduced
by Kleinrock \cite{Kle67}. Under the DPS policy jobs are organized
in classes, which share a single server. The capacity that each
class obtains depends on the number of jobs currently presented in
all classes. All jobs present in the system are served
simultaneously at rates controlled by the vector of weights
$g_k>0$, $k=1,\ldots,M\}$, where $M$ is the number of classes. If
there are $N_j$ jobs in  class $j$, then each job of this class is
served with the rate $g_j/\sum_{k=1}^M{g_k N_k}$. When all weights
are equal, DPS system is equivalent to the standard PS policy.

The DPS policy model has recently received a lot of attention due
to its wide range of application. For example, DPS could be
applied to model flow level sharing of TCP flows with different
flow characteristics such as different RTTs and packet loss
probabilities. DPS also provides a natural approach to model the
weighted round-robin discipline, which is used in operating
systems for task scheduling. In the Internet one can imagine the
situation that servers provide different service according to the
payment rates. For more applications of DPS in communication
networks see \cite{AJK04}, \cite{BT01}, \cite{CBBLR05},
\cite{GM02}, \cite{HT05}.

Varying DPS weights it is possible to give priority to different
classes at the expense of others, control their instantaneous
service rates and optimize different system characteristics as
mean sojourn time and so on. So, the proper weight selection is an
important task, which is not easy to solve because of the model's
complexity.

The previously obtained results on DPS model are the following.
Kleinrock  in {\cite{Kle67}} was first studying DPS. Then the
paper of Fayolle et al. {\cite{FMI90}} provided results for the
DPS model. For the exponentially distributed required service
times the authors obtained the expression of the expected sojourn
time as a solution of a system of linear equations. The authors
show that independently of the weights the slowdown for the
expected conditional response time under the DPS policy tends to
the constant slowdown of the PS policy as the service requirements
increases to infinity.

Rege and Sengupta in {\cite{RS94}} %decomposition theorem
proved a decomposition theorem for the conditional sojourn time.
For exponential service time distributions in {\cite{RS96}} they
obtained higher moments of the queue length distribution as the
solutions of linear equations system and also provided a theorem
for the heavy-traffic regime. Van Kessel et al. in \cite{KNB05},
\cite{KK06} study the performance of DPS in an asymptotic regime
using time scaling. For general distributions of the required
service times the approximation analysis was carried out by Guo
and Matta in {\cite{GM02}}.  Altman et al. {\cite{AJK04}} study
the behavior of the DPS policy in overload. Most of the results
obtained for the DPS queue were collected together in the survey
paper of Altman et al. {\cite{AAA06}}.

Avrachenkov et al. in {\cite{AABN}} % compare the PS and DPS policies for the
%exponential required service time distributions.
proved that the mean unconditional response time of each class is
finite under the usual stability condition. They determine the
asymptote of the conditional sojourn time for each class assuming
finite service time distribution with finite variance.

The problem of weights selection in the DPS policy when the job
size distributions are exponential was studied by Avrachenkov et
al. in {\cite{AABN}} and by Kim and Kim in {\cite{KK06}}. In
{\cite{KK06}} it was shown that the DPS policy reduces the
expected sojourn time in comparison with PS policy when the
weights increase in the opposite order with the means of job
classes. Also in {\cite{KK06}} the authors formulate a conjecture
about the monotonicity of the expected sojourn time of the DPS
policy. The idea of conjecture is that comparing two DPS policies,
one which has a weight vector closer to the optimal policy provided by $c\mu$-rule, see {\cite{RR94}}, 
has smaller expected sojourn time. Using the method
described in {{\cite{KK06}}} in  the present paper we prove this
conjecture with some restrictions on the system parameters. The
restrictions on the system are such that the result is true for
systems for which the values of the job size distribution means
are very different from each other.  The restriction can be
overcome by setting the same weights for the classes, which have
similar means. The condition on means is a sufficient, but not a
necessary condition. It becomes less strict when the system is
less loaded.

The paper is organized as follows. In Section~{\ref{sec:dps:def}} we give general
definitions of the DPS policy and formulate the problem of expected sojourn time minimization. In
Section~{\ref{sec:dps:monot}} we formulate the main Theorem and prove it.
In Section~{\ref{sec:dps:numer_res}} we give the numerical results. Some technical
proofs can be found in the Appendix.

\section{Previous results and problem formulation}{\label{sec:dps:def}}

We consider the Discriminatory Processor Sharing (DPS) model. All
jobs are organized in $M$ classes and share a single server. Jobs
of class $k=1,\ldots,M$ arrive with a Poisson process with rate
$\lambda_k$ and have required service-time distribution
$F_k(x)=1-e^{-\mu_k x}$ with mean $1/\mu_k$. The load of the
system is $\rho=\sum_{k=1}^M{\rho_k}$ and
$\rho_k=\lambda_k/\mu_k$, $k=1,\ldots,M$. We consider that the
system is stable, $\rho<1$. Let us denote
$\lambda=\sum_{k=1}^{M}{\lambda_k}$.

The state of the system is controlled by a vector of weights
$g=(g_1,\ldots,g_M)$, which denotes the priority for the job
classes. If in the class $k$ there are currently $N_k$ jobs, then
each job of class $k$ is served with the rate equal to
$g_j/\sum_{k=1}^M{g_k N_k}$, which depends on the current system
state, or on the number of jobs in each class.

Let $\overline{T}^{DPS}$ be the expected sojourn time of the DPS
system. We have
\begin{eqnarray*}
&& \overline{T}^{DPS} = \sum_{k=1}^{M}
{\frac{\lambda_k}{\lambda}\overline{T}_k},
\end{eqnarray*}
where $\overline{T}_k$ are expected sojourn times for class $k$.
The expressions for the expected sojourn times $\overline{T}_k$,
$k=1,\ldots,M$ can be found as a solution of the system of linear
equations, see {\cite{FMI90}},
\begin{eqnarray}{\label{eq:dps:T_k_system}}
\overline{T}_k \left( 1-\sum_{j=1}^{M}{\frac{\lambda_j g_j}{\mu_j
g_j+\mu_kg_k}}\right)-\sum_{j=1}^{M}{\frac{\lambda_j g_j
\overline{T}_j }{\mu_j g_j+\mu_k g_k}} = \frac{1}{\mu_k}, \quad
k=1,\ldots,M.
\end{eqnarray}

Let us notice that for the standard Processor Sharing system
\begin{eqnarray*}
&& \overline{T}^{PS} = \frac{m}{1-\rho}.
\end{eqnarray*}

One of the problems when studying DPS is to minimize the expected
sojourn time $\overline{T}^{DPS}$ with some weight selection.
Namely, find $g^*$ such as
\begin{eqnarray*}
&& \overline{T}^{DPS}(g^*)=\min_g{\overline{T}^{DPS}(g)}.
\end{eqnarray*}
This is a general problem and to simplify it the following subcase
is considered. To find a set $G$ such that
\begin{eqnarray}
&& \overline{T}^{DPS}(g^*) \leq {\overline{T}^{PS}} , \quad
\forall g^* \in G. \label{dps:eq:dop:probl_dps_ps}
\end{eqnarray}
For the case when job size distributions are exponential the solution of
({\ref{dps:eq:dop:probl_dps_ps}}) is given by Kim and Kim in
{\cite{KK06}} and is as follows. If the means of the classes are
such as $ \mu_1 \geq \mu_2 \geq \ldots \geq \mu_M$, then $G$ consists
of all such vectors which satisfy
\begin{eqnarray*}
&&  G=\{g|\,\,g_1 \geq g_2\geq \ldots \geq g_M\}.
\end{eqnarray*}
Using the approach of {\cite{KK06}} we solve more general problem
about the monotonicity of the expected sojourn time in the DPS
system, which we formulate in the following section as
Theorem~{\ref{theorem:dps:T_monot}}.

\section{Expected sojourn time monotonicity}{\label{sec:dps:monot}}

Let us formulate and prove the following Theorem.

\begin{Theorem}{\label{theorem:dps:T_monot}}
Let the job size distribution for every class be exponential with
mean $\mu_i$, $i=1,\ldots, M$ and we enumerate them in the
following way
\begin{eqnarray*}
&& \mu_1 \geq \mu_2 \geq \ldots \geq \mu_M.
\end{eqnarray*}
Let us consider two different weight policies for the DPS system,
which we denote as $\alpha$ and $\beta$. Let $\alpha, \beta \in
G$, or
\begin{eqnarray*}
&& \alpha_1 \geq \alpha_2 \geq \ldots \geq \alpha_M, \\
&& \beta_1 \geq \beta_2 \geq \ldots \geq \beta_M.
\end{eqnarray*}
The expected sojourn time of the DPS policies with weight vectors $\alpha$ and $\beta$ satisfies
\begin{eqnarray}
&& \overline{T}^{DPS}(\alpha) \leq \overline{T}^{DPS}(\beta),
\label{eq:dps:T_alpha_beta_compar}
\end{eqnarray}
if the weights $\alpha$ and $\beta$ are such that:
\begin{eqnarray}
&& \frac{\alpha_{i+1}}{\alpha_i} \leq \frac{\beta_{i+1}}{\beta_i},
\quad i=1,\ldots,M-1,  \label{eq:dps:alpha_beta_compar}
\end{eqnarray}
and the following restriction is satisfied:
\begin{eqnarray}
&& \quad \frac{\mu_{j+1}}{\mu_j} \leq  1-\rho,  \label{eq:dps:condit_satisf}
\end{eqnarray}
for every $j=1,\ldots,M$.
\end{Theorem}

\begin{Remark} {\label{remark:dps:weight_sel}}
If for some classes $j$ and $j+1$ condition
(\ref{eq:dps:condit_satisf}) is not satisfied, then in practice,
by choosing the weights of these classes to be equal, we can still
use Theorem~{\ref{theorem:dps:T_monot}}. Namely, for classes such
as $\frac{\mu_{j+1}}{\mu_j} > 1-\rho$, we suggest to set
$\alpha_{j+1}=\alpha_j$ and $\beta_{j+1}=\beta_j$.
\end{Remark}

\begin{Remark}
Theorem~{\ref{theorem:dps:T_monot}} shows that the expected
sojourn time $\overline{T}^{DPS}(g)$ is monotonous according to
the selection of weight vector $g$. The closer is the weight
vector to the optimal policy, provided by $c \mu$-rule, the smaller is the expected
sojourn time. This is shown by the condition
({\ref{eq:dps:alpha_beta_compar}}), which shows that vector
$\alpha$ is closer to the optimal $c \mu$-rule policy then vector
$\beta$. %Theorem~{\ref{theorem:dps:T_monot}} proves that when the
%weigh policy goes to the optimal strict priority policy, the expected
%sojourn time decreases.

Theorem~{{\ref{theorem:dps:T_monot}}} is proved with restriction
({\ref{eq:dps:condit_satisf}}). This restriction is a sufficient
and not a necessary condition on system parameters. It shows that
the means of the job classes have to be quite different from each
other. This restriction can be overcome, giving the same weights
to the job classes, which mean values are similar. Condition
({\ref{eq:dps:condit_satisf}}) becomes less strict as the system
becomes less loaded.
\end{Remark}

To prove Theorem~{\ref{theorem:dps:T_monot}} let us first give
some notations and prove additional Lemmas.

Let us rewrite linear system ({\ref{eq:dps:T_k_system}}) in the
matrix form. Let
$\overline{T}^{(g)}=[\overline{T}_1^{(g)},\ldots,\overline{T}_M^{(g)}]^{T}$
be the vector of $\overline{T}_k^{(g)}$, $k=1,\ldots,M$. Here by
$[\mbox{\,\,}]^{T}$ we mean transpose sign, so $[\mbox{\,\,}]^{T}$
is a vector. By $[\mbox{\,\,}]^{(g)}$ we note that this element
depends on the weight vector selection $g\in G$. Let us consider
that later in the paper vectors $g, \alpha, \beta \in G$, if the
opposite is not noticed.  Let define matrices $A^{(g)}$ and
$D^{(g)}$ in the following way.
\begin{equation}
A^{(g)} =\left(\\
\begin{array}{l l l l l}
\frac{\lambda_1 g_1}{\mu_1 g_1 +\mu_1 g_1} & \frac{\lambda_2 g_2}{\mu_1 g_1 +\mu_2 g_2} & \ldots& \frac{\lambda_M g_M}{\mu_1 g_1 +\mu_M g_M} \\
\frac{\lambda_1 g_1}{\mu_2 g_2 +\mu_1 g_1} & \frac{\lambda_2 g_2}{\mu_2 g_2 +\mu_2 g_2} & \ldots& \frac{\lambda_M g_M}{\mu_2 g_2 +\mu_M g_M} \\
\ldots\\
\frac{\lambda_1 g_1}{\mu_M g_M +\mu_1 g_1} & \frac{\lambda_2 g_2}{\mu_M g_M +\mu_2 g_2} & \ldots& \frac{\lambda_M g_M}{\mu_M g_M +\mu_M g_M} \\
\end{array}\right) \label{eq:dps:A_matr_general}
\end{equation}
\begin{equation}
D^{(g)} =\left(\\
\begin{array}{l l l l l}
\sum_i{\frac{\lambda_i g_i}{\mu_1 g_1 +\mu_i g_i}}  & 0 & \ldots & 0 \\
0 & \sum_i{\frac{\lambda_i g_i}{\mu_2 g_2 +\mu_i g_i}} & \ldots & 0 \\
\ldots\\
0 & 0 & \ldots & \sum_i{\frac{\lambda_i g_i}{\mu_M g_M +\mu_i g_i}} \\
\end{array}\right)\label{eq:dps:D_matr_general}
\end{equation}
Then ({\ref{eq:dps:T_k_system}})  becomes
\begin{eqnarray}
&& (E-D^{(g)}-A^{(g)}) \overline{T}^{(g)}
=\left[\frac{1}{\mu_1}\ldots.\frac{1}{\mu_M}\right]^{T}.
\label{eq:dps:system_matrix_form}
\end{eqnarray}
We need to find the expected sojourn time of the DPS system
$\overline{T}^{DPS}(g)$. According to the definition of
$\overline{T}^{DPS}(g)$ and equation
(\ref{eq:dps:system_matrix_form}) we have
\begin{eqnarray}
&& \overline{T}^{DPS}(g) =
\frac{1}{\lambda}[\lambda_1,\ldots,\lambda_M] \overline{T}^{(g)} =
\frac{1}{\lambda}[\lambda_1,\ldots,\lambda_M]
(E-D^{(g)}-A^{(g)})^{-1}
\left[\frac{1}{\mu_1},\ldots,\frac{1}{\mu_M}\right]^{T}.
\label{eq:dps:T_mart_form}
\end{eqnarray}

Let us consider the case when $\lambda_i=1$ for $i=1,\ldots,M$.
This results can be extended for the case when $\lambda_i$ are
different, we prove it following the approach of {\cite{KK06}} in
Proposition~{\ref{propos:dps:extension}} at the end of the current
Section. Equation (\ref{eq:dps:T_mart_form}) becomes
\begin{eqnarray}
\overline{T}^{DPS}(g) & = &\underline{1}'(E-D^{(g)}-A^{(g)})^{-1}
\left[{\rho_1},\ldots,\rho_M\right]^{T}{\lambda^{-1}}.
\label{eq:dps:T_matr_form_lambda_eq}
\end{eqnarray}
Let us give the following notations.
\begin{eqnarray*}
&& \sigma_{ij}^{(g)}=\frac{g_j}{\mu_i g_i+\mu_j g_j}.
\end{eqnarray*}
Then $\sigma_{ij}^{(g)}$ have the following properties.
\begin{Lemma} {\label{lemma:dps:lemma_sigma_ij}}
 $\sigma_{ij}^{(g)}$ and $\sigma_{ji}^{(g)}$ satisfy
\begin{eqnarray}
&& \sigma_{ij}^{(g)} g_i=\sigma_{ji}^{(g)}{g_j}, \nonumber\\
&& \frac{\sigma_{ij}^{(g)}} {\mu_i} +
\frac{\sigma_{ji}^{(g)}}{\mu_j} =\frac{1}{\mu_i\mu_j}.
% \quad %{\sigma_{ij}^{(g)}} {\mu_j}+{\sigma_{ji}^{(g)}} {\mu_i}=1.
\label{eq:dps:eq:sigma_propert}
\end{eqnarray}
\end{Lemma}
\begin{proof} Follows from the definition of $\sigma_{ij}^{(g)}$.
\end{proof}
Then matrices $A^{(g)}$ and $D^{(g)}$  given by
(\ref{eq:dps:A_matr_general}) and (\ref{eq:dps:D_matr_general})
can be rewritten in the terms of $\sigma_{ij}^{(g)}$.
\begin{eqnarray*}
&& A_{i,j}^{(g)}=\sigma_{ij}^{(g)}, \quad  i,j = 1,\ldots,M,  \\
&& D_{i,i}^{(g)} = \sum_j{\sigma_{ij}^{(g)}}, \quad i=1,\ldots,M, \\
&& D_{i,j}^{(g)}=0, \quad i,j=1,\ldots,M, \quad i\neq j.
\end{eqnarray*}
For weight vectors $\alpha, \beta$ the following Lemma is true.
\begin{Lemma}{\label{lemma:dps:weights_compar}}
If $\alpha$ and $\beta$ satisfy (\ref{eq:dps:alpha_beta_compar}),
then
\begin{eqnarray}{\label{eq:dps:alpha_beta_general_compar}}
&& \frac{\alpha_j}{\alpha_i} \leq \frac{\beta_j}{\beta_i}, \quad
i=1,\ldots,M-1, \quad \forall \, j\geq i.
\end{eqnarray}
\end{Lemma}
\begin{proof} Let us notice that if $a<b$ and $c<d$, then $ac<bd$ when $a,b,c,d$ are
positive. Also if $j>i$ then there exist such $l>0$ that $j=i+l$.
Then
\begin{eqnarray*}
&& \frac{\alpha_{i+1}}{\alpha_i} \leq \frac{\beta_{i+1}}{\beta_i},
 \quad \frac{\alpha_{i+2}}{\alpha_{i+1}} \leq
\frac{\beta_{i+2}}{\beta_{i+1}}, \quad \ldots \quad
\frac{\alpha_{i+l}}{\alpha_{i+l-1}} \leq
\frac{\beta_{i+l}}{\beta_{i+l-1}}, \quad i=1,\ldots,M-2.
\end{eqnarray*}
Multiplying left and right parts of the previous inequalities we
get the following:
\begin{eqnarray*}
\frac{\alpha_{i+l}}{\alpha_{i}} \leq
\frac{\beta_{i+l}}{\beta_{i}}, \quad i=1,\ldots,M-2,
\end{eqnarray*}
which proves Lemma~{\ref{lemma:dps:weights_compar}}.
\end{proof}

\begin{Lemma}\label{lemma:dps:sigma_prop}
If $\alpha$ and $\beta$ satisfy
({\ref{eq:dps:alpha_beta_general_compar}}), then
\begin{eqnarray*}
&& \sigma_{ij}^{(\alpha)} \leq \sigma_{ij}^{(\beta)}, \quad i \leq j,\\
&& \sigma_{ij}^{(\alpha)} \geq \sigma_{ij}^{(\beta)}, \quad i \geq
j.
\end{eqnarray*}
\end{Lemma}
\begin{proof} As ({\ref{eq:dps:alpha_beta_general_compar}}) then
\begin{eqnarray*}
&& \frac{\alpha_j}{\alpha_i}\leq\frac{ \beta_j}{\beta_i}, \quad i
\leq j, \\
&& \alpha_j \mu_i\beta_i  \leq \beta_j\mu_i\alpha_i, \quad i
\leq j,  \\
&& {\alpha_j}({\mu_i\beta_i+\mu_j\beta_j) \leq
\beta_j({\mu_i\alpha_i+\mu_j\alpha_j})}, \quad i
\leq j,  \\
&& \frac{\alpha_j}{\mu_i\alpha_i+\mu_j\alpha_j} \leq
\frac{\beta_j}{\mu_i\beta_i+\mu_j\beta_j} , \quad i
\leq j,  \\
&& \sigma_{ij}^{(\alpha)} \leq \sigma_{ij}^{(\beta)}, \quad i \leq
j.
\end{eqnarray*}
We prove the second inequality of
Lemma~{\ref{lemma:dps:sigma_prop}} in a similar way.
\end{proof}

\begin{Lemma}{\label{lemma:dps:T_monot_y_dep}}
If $\alpha$, $\beta$ satisfy (\ref{eq:dps:alpha_beta_compar}),
then
\begin{eqnarray*}
&& \overline{T}^{DPS}(\alpha) \leq \overline{T}^{DPS}(\beta),
\end{eqnarray*}
when the elements of vector $y=\underline{1}'
(E-B^{(\alpha)})^{-1}M$ are such that $y_1 \geq y_2 \geq \ldots
\geq y_M$.
\end{Lemma}
\begin{proof}
Let us denote $B^{(g)}=A^{(g)}+D^{(g)}$, $g=\alpha, \beta$. Then
as (\ref{eq:dps:T_matr_form_lambda_eq})
\begin{eqnarray*} &&
\overline{T}^{DPS}(g)={\lambda^{-1}}\underline{1}'(E-B^{(g)})^{-1}
\left[\rho_1,\ldots,\rho_M\right]^{T}, \quad g=\alpha,\beta.
\end{eqnarray*}
Following the method described in {\cite{KK06}} we get the
following.
\begin{eqnarray*}
\overline{T}^{DPS}(\alpha)-\overline{T}^{DPS}(\beta)&=&
{\lambda^{-1}}\underline{1}'(E-B^{(\alpha)})^{-1}
\left[\rho_1,\ldots,\rho_M\right]^{T}-{\lambda^{-1}}\underline{1}'(E-B^{(\beta)})^{-1}
\left[\rho_1,\ldots,\rho_M\right]^{T}=\\
& = & {\lambda^{-1}}\underline{1}'((E-B^{(\alpha)})^{-1}
-(E-B^{(\beta)})^{-1}) \left[\rho_1,\ldots,\rho_M\right]^{T}=\\
& = & {\lambda^{-1}}\underline{1}'
((E-B^{(\alpha)})^{-1}(B^{(\alpha)}-B^{(\beta)})(E-B^{(\beta)})^{-1})
\left[\rho_1,\ldots,\rho_M\right]^{T}.
\end{eqnarray*}
Let us denote $M$ as a diagonal matrix $M=diag(\mu_1,\ldots,\mu_M)$
and
\begin{eqnarray}
&& y=\underline{1}' (E-B^{(\alpha)})^{-1}M. \label{eq:dps:y_def}
\end{eqnarray}
Then
\begin{eqnarray*}
\overline{T}^{DPS}(\alpha)-\overline{T}^{DPS}(\beta) &=&
\underline{1}' (E-B^{(\alpha)})^{-1}MM^{-1}(B^{(\alpha)}-B^{(\beta)})\overline{T}^{(\beta)}=\\
&=& y M^{-1}(B^{(\alpha)}-B^{(\beta)}) \overline{T}^{(\beta)}=\\
&=& \sum_{i,j}\left({\frac{y_j}{\mu_j}\sigma_{ji}^{(\alpha)}
+\frac{y_i}{\mu_i}\sigma_{ij}^{(\alpha)} -
\left(\frac{y_j}{\mu_j}\sigma_{ji}^{(\beta)}
+\frac{y_i}{\mu_i}\sigma_{ij}^{(\beta)}\right)}\right)\overline{T}_j^{(\beta)}=\\
&=& \sum_{i,j}\left({y_j}\left(
\frac{\sigma_{ji}^{(\alpha)}}{\mu_j}
-\frac{\sigma_{ji}^{(\beta)}}{\mu_j} \right)
+\frac{y_i}{\mu_i}(\sigma_{ij}^{(\alpha)}-\sigma_{ij}^{(\beta)})\right)\overline{T}_j^{(\beta)}.
\end{eqnarray*}
As ({\ref{eq:dps:eq:sigma_propert}}):
\begin{eqnarray*}
&& \frac{\sigma_{ji}^{(g)}}{\mu_j}=\frac{1}{\mu_i \mu_j}-
\frac{\sigma_{ij}^{(g)}}{\mu_i},\quad g=\alpha, \beta,
\end{eqnarray*}
then
\begin{eqnarray*}
\overline{T}^{DPS}(\alpha)-\overline{T}^{DPS}(\beta)  &=&
\sum_{i,j}\left(-y_j\left( \frac{\sigma_{ij}^{(\alpha)}}{\mu_i}
-\frac{\sigma_{ij}^{(\beta)}}{\mu_i}\right)
+\frac{y_i}{\mu_i}(\sigma_{ij}^{(\alpha)}-\sigma_{ij}^{(\beta)})\right)\overline{T}_j^{(\beta)}=\\
&=& \sum_{i,j}\left(-\frac{y_j}{\mu_i}
\left(\sigma_{ij}^{(\alpha)}-\sigma_{ij}^{(\beta)}\right)
+\frac{y_i}{\mu_i}(\sigma_{ij}^{(\alpha)}-\sigma_{ij}^{(\beta)})\right)\overline{T}_j^{(\beta)}=\\
&=& \sum_{i,j}\left
(\left(\sigma_{ij}^{(\alpha)}-\sigma_{ij}^{(\beta)}\right)(y_i-y_j)
\frac{1}{\mu_i}\right)\overline{T}_j^{(\beta)}.
\end{eqnarray*}
Using Lemma {\ref{lemma:dps:sigma_prop}} we get
$\left(\sigma_{ij}^{(1)} -\sigma_{ij}^{(2)}\right)(y_i-y_j)$ is
negative for $i,j=1,\ldots,M$ when $y_1 \geq y_2 \geq \ldots \geq y_M$.
This proves the statement of Lemma~{\ref{lemma:dps:T_monot_y_dep}}.
\end{proof}

\begin{Lemma}{\label{lemma:dps:y_inequality}}
Vector $y$ given by (\ref{eq:dps:y_def}) satisfies
\begin{eqnarray*}
&& y_1 \geq y_2 \geq \ldots \geq y_M,
\end{eqnarray*}
if the following is true:
\begin{eqnarray*}
&&  \frac{\mu_{i+1}}{\mu_{i}} \leq 1-\rho,
\end{eqnarray*}
for every $i=1,\ldots,M$.
\end{Lemma}
\begin{proof}
The proof could be found in the appendix.
\end{proof}

\begin{Remark}\label{remark:dps:weight_sel_for_y}
For the job classes such as $\frac{\mu_{i+1}}{\mu_{i}} > 1-\rho$
we prove that to make $y_i \geq y_{i+1}$ it is sufficient to set
the weights of these classes equal, ${\alpha_{i+1}=\alpha_i}$.
\end{Remark}

Combining the results of Lemmas~{\ref{lemma:dps:weights_compar}},
{\ref{lemma:dps:sigma_prop}}, {\ref{lemma:dps:T_monot_y_dep}} and
{\ref{lemma:dps:y_inequality}} we prove the statement of the
Theorem~{\ref{theorem:dps:T_monot}}.
Remark~\ref{remark:dps:weight_sel_for_y} gives the
Remark~{\ref{remark:dps:weight_sel}}  after
Theorem~{\ref{theorem:dps:T_monot}}.

\begin{Proposition}
%{Additional remarks}
{\label{propos:dps:extension}}
The result of  Theorem~{\ref{theorem:dps:T_monot}} is
extended to the case when $\lambda_i \neq 1$.
\end{Proposition}
\begin{proof}
Let us first consider the case when all $\lambda_i=q$, $i=1,\ldots,M$. It can be shown that
for this case the proof of
Theorem~{\ref{theorem:dps:T_monot}} is equivalent to the proof of
the same Theorem but for the new system with
$\lambda_i^{*}=1$,  $\mu_i^{*}=q \mu_i$, $i=1,\ldots,M$.
For this new system the results of Theorem~{\ref{theorem:dps:T_monot}} is
evidently true and restriction
({\ref{eq:dps:condit_satisf}}) is not changed. Then,
Theorem~{\ref{theorem:dps:T_monot}} is true for the initial system
as well.

If $\lambda_i$ are rational, then they could be written in
$\lambda_i=\frac{p_i}{q}$, where $p_i$ and $q$ are positive
integers. Then each class can be presented as $p_i$ classes with
equal means $1/\mu_i$ and intensity $1/q$. So, the DPS system can
be considered as the DPS system with $p_1+\ldots+p_K$ classes with
the same arrival rates $1/q$. The result of
Theorem~{\ref{theorem:dps:T_monot}} is extended on this case.

If $\lambda_i$, $i=1,\ldots,M$ are positive and real we apply the previous case of
rational $\lambda_i$ and use continuity.
\end{proof}

\section{Numerical results} {\label{sec:dps:numer_res}}

Let us consider a DPS system with $3$ classes. Let us consider the
set of normalized weights vectors $g(x)=(g_1(x),g_2(x),g_3(x))$ ,
$\sum_{i=1}^3{g_i(x)}=1$, $g_i(x)
={{x^{-i}/(\sum_{i=1}^3{{x^{-i}}}})}$, $x>1$. Every point $x>1$
denotes a weight vector. Vectors $g(x), g(y)$ satisfy property
(\ref{eq:dps:alpha_beta_compar}) when $1<y \leq x$, namely $
{g_{i+1}(x)}/{g_i(x)} \leq {g_{i+1}(y)}/{g_i(y)}$, $i=1,2$, $1<y
\leq x$. On Figures~{\ref{fig:compar_weights_sat}},~{\ref{fig:compar_weights_not_sat}} we plot
$\overline{T}^{DPS}$ with weights vectors $g(x)$ as a function of
$x$, the expected sojourn times $\overline{T}^{PS}$ for the PS policy and
$\overline{T}^{opt}$ for the optimal $c \mu$-rule policy.

On Figure~{\ref{fig:compar_weights_sat}} we plot the expected sojourn time for the 
case when condition  ({\ref{eq:dps:condit_satisf}}) is satisfied for three
classes. The parameters are: $\lambda_i=1$, $i=1,2,3$,
$\mu_1=160$, $\mu_2=14$, $\mu_3=1.2$, then $\rho=0.911$.
On Figure~{\ref{fig:compar_weights_not_sat}} we plot the expected sojourn time for the 
case when condition  ({\ref{eq:dps:condit_satisf}}) is \textit{not} satisfied for three
classes. The parameters are: $\lambda_i=1$, $i=1,2,3$,
$\mu_1=3.5$, $\mu_2=3.2$, $\mu_3=3.1$, then $\rho=0.92$. 
One can see that $\overline{T}^{DPS}(g(x)) \leq \overline{T}^{DPS}(g(y))$, $1<y \leq x$ even when the restriction ({\ref{eq:dps:condit_satisf}}) is not satisfied. 

 \begin{figure}
   \begin{minipage}{0.48\linewidth}
        \centering {\epsfxsize=2.5 in \epsfbox{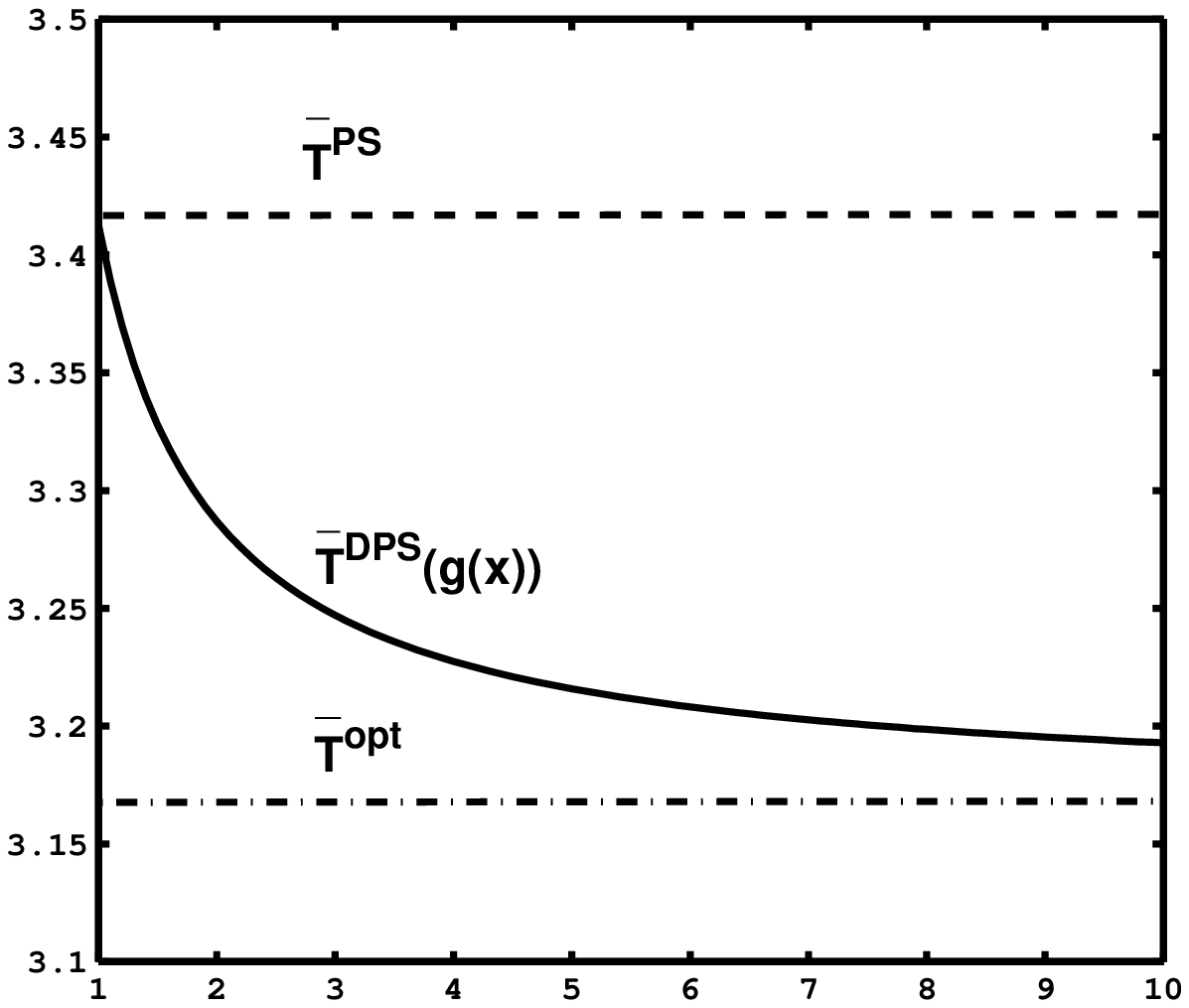}}
        \caption{$\overline{T}^{DPS}(g(x))$, $\overline{T}^{PS}$,  $\overline{T}^{opt}$ functions, condition satisfied.}{\label{fig:compar_weights_sat}}
   \end{minipage}
   \hfill
   \begin{minipage}{0.48\linewidth}
        \centering {\epsfxsize=2.5 in \epsfbox{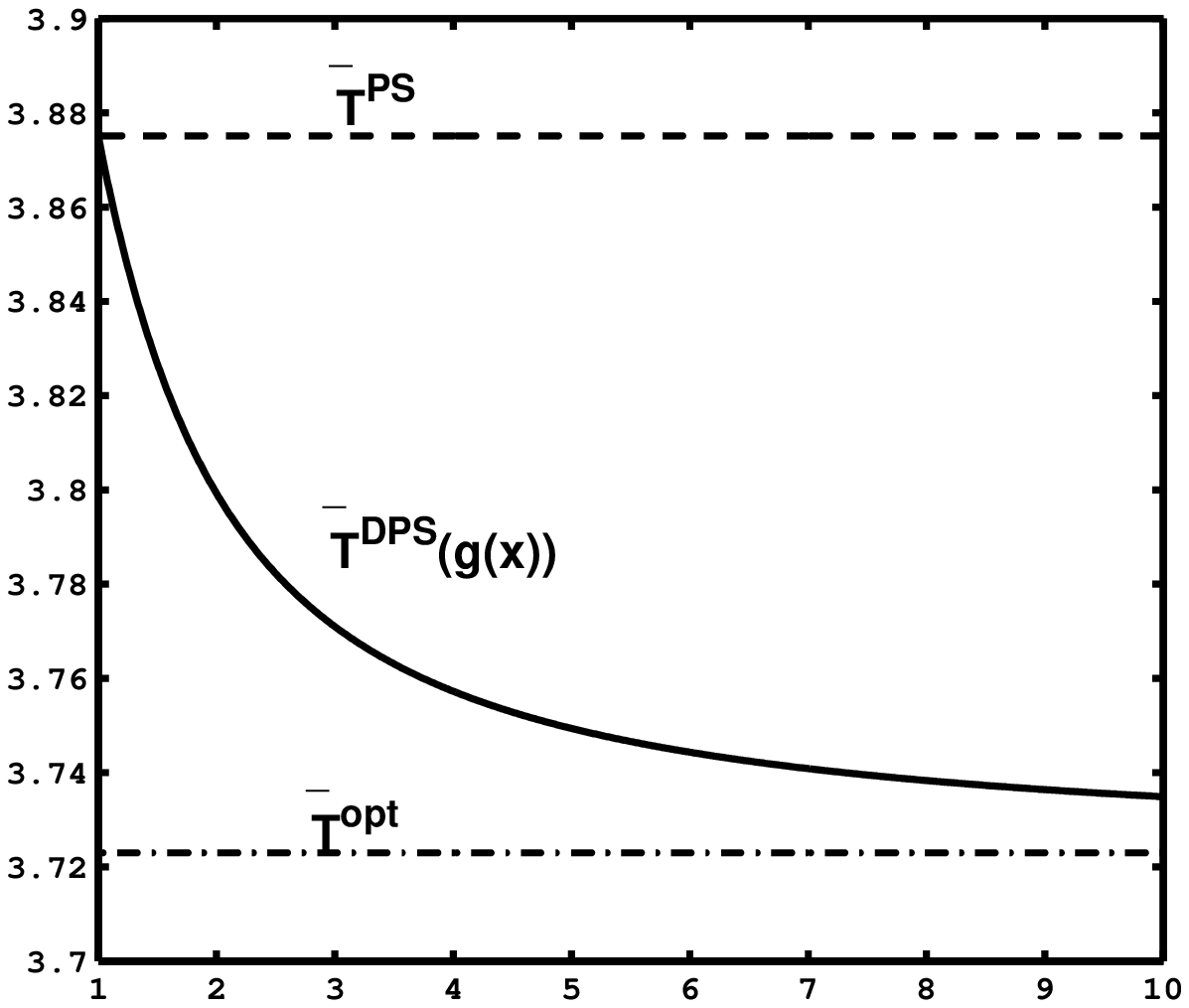}}
 \caption{$\overline{T}^{DPS}(g(x))$, $\overline{T}^{PS}$,  $\overline{T}^{opt}$ functions,  condition not satisfied}{\label{fig:compar_weights_not_sat}}
   \end{minipage}
 \end{figure}

\section{Conclusion}

We study the DPS policy with exponential job size distributions.
One of the main problems studying DPS is the expected sojourn time
minimization according to the weighs selection. In the present
paper we compare two DPS policies with different weights. We show
that the expected sojourn time is smaller for the policy with the
weigh vector closer to the optimal policy vector, provided by $c \mu$-rule.  
So, we prove the monotonicity of the expected sojourn
time for the DPS policy according to the weight vector selection.

The result is proved with some restrictions on system parameters.
The found restrictions on the system parameters are such that the
result is true for systems such as the mean values of the job
class  size distributions are very different from each other. We
found, that to prove the main result it is sufficient to give the
same weights to the classes with similar means. The found
restriction is a sufficient and not a necessary condition on a
system parameters. When the load of the system decreases, the
condition becomes less strict.

\section*{Acknowledgment}
I would like to thank K. Avrachenkov, P. Brown and U. Ayesta for
fruitful discussions and suggestions.

\section{Appendix}

In the following proof in the notations we do not use  the
dependency of the parameters on $g$ to simplify the notations. We
consider that vector $g \in G$, or $g_1 \geq g_2 \ldots \geq g_M$.
To simplify the notations let us use $\sum_k$ instead of $\sum_{k=1}^M$.

\textbf{Lemma~{\ref{lemma:dps:y_inequality}}.} Vector
$y=\underline{1}' (E-B)^{-1}M $ satisfies
\begin{eqnarray*}
&& y_1 \geq y_2 \geq \ldots \geq y_M,
\end{eqnarray*}
if the following is true:
\begin{eqnarray*}
&&  \frac{\mu_{i+1}}{\mu_{i}} \leq 1-\rho,
\end{eqnarray*}
for every $i=1,\ldots,M$.
\begin{proof} Using the results of the following Lemmas we prove
the statement of Lemma~{\ref{lemma:dps:y_inequality}} and give the
proof for Remark~\ref{remark:dps:weight_sel_for_y}.
\end{proof}

%\begin{Remark}
%For the job classes such as $\frac{\mu_{i+1}}{\mu_{i}} > 1-\rho$ we prove that to keep $y_j \geq y_{j+1}$ it is sufficient to set the weights
%of these classes equal, ${\alpha_{i+1}=\alpha_i}$.
%\end{Remark}

Let us give the following notations
\begin{eqnarray}
&& \tilde{\mu}={\mu}^{T}(E-D)^{-1}, \label{eq:dps:mu_tilde_def}\\
&& \tilde{A} = M^{-1}A M (E-D)^{-1}.\label{eq:dps:A_tilde_def}
\end{eqnarray}

Let us notice the following
\begin{eqnarray*}
(E-D)^{-1}_j &  = & \frac{1}{{1-\sum_{k} {\frac{g_k}{\mu_j g_j+
\mu_k g_k}}}} =\frac{1}{{1-\rho + \sum_{k} {\frac{\mu_j
g_j}{\mu_k(\mu_j g_j+\mu_k g_k)}}}} >0, \quad  j=1,\ldots,M, \\
{\tilde{A}_{ij}} & = & \frac{{\frac{ \mu_j g_j}{\mu_i(\mu_i g_i+
\mu_j g_j)}}}{1-\sum_{k} {\frac{g_k}{\mu_j g_j+ \mu_k g_k}}} =
\frac{{\frac{ \mu_j g_j}{\mu_i(\mu_i g_i+ \mu_j g_j)}}} {{{1-\rho
+ \sum_{k}{\frac{\mu_j g_j}{\mu_k(\mu_j g_j+\mu_k g_k)}}}}} > 0,
\quad i,j=1,\ldots,M
\end{eqnarray*}
Let us give the following notation
\begin{eqnarray*}
&& f(x) = \sum_{k}{\frac{x}{\mu_k(x+\mu_k g_k)}}.
\end{eqnarray*}
Then
\begin{eqnarray*}
(E-D)^{-1}_j & =& \frac{1}{{1-\rho + f(\mu_j g_j)}}, \quad
j=1,\ldots,M ,\\
{\tilde{A}_{ij}} &=& {\frac{ \mu_j g_j}{\mu_i(\mu_i g_i+ \mu_j
g_j)(1-\rho + f(\mu_j g_j))}}, \quad i,j=1,\ldots,M.
\end{eqnarray*}
Let us first prove additional Lemma.
\begin{Lemma}{\label{lemma:dps:tilde_A_contraction}}
Matrix
\begin{eqnarray*}
&& \tilde{A}=M^{-1}A M (E-D)^{-1}
\end{eqnarray*}
is a positive contraction.
\end{Lemma}
\begin{proof}
Matrix $\tilde{A}$ is a positive operator as elements of matrices
$M$ and $A$ are positive and elements of matrix $(E-D)^{-1}$ are
positive. Let $\Omega=\{X | x_i\geq0, \,\, i=1,\ldots,M\}$. If $X
\in \Omega$, then $\tilde{A} X \in \Omega$. Then to prove that
matrix $\tilde{A}$ is a contraction it is enough to show that
\begin{eqnarray} {\label{eq:dps:contraction_condit}}
&& \exists \, q,\quad  0<q<1, \quad ||\tilde{A}X|| \leq q ||X||,
\quad \forall\,\, X \in \Omega.
\end{eqnarray}
As $X \in \Omega$, then we can take
$||X||=\underline{1}'X=\sum_{i}x_i$. Then
\begin{eqnarray*}
\underline{1}' \tilde{A} X &=& \sum_{j} x_j \sum_{i}
\tilde{A}_{ij} =\sum_{j} x_j \frac{\sum_{i} {\frac{ \mu_j
g_j} {\mu_i( \mu_j g_j+ \mu_i g_i)}}}{(1-\rho+g(\mu_j g_j))} = \\
& =& \sum_j x_j \frac{f(\mu_j g_j)}{1-\rho +f(\mu_j g_j)} =\sum_j
x_j \left(1-\frac{1-\rho}{1-\rho +f(\mu_j g_j)}\right)=\\
&=& \sum_j x_j - (1-\rho) \sum_j \frac{x_j}{1-\rho +f(\mu_j g_j)}.
\end{eqnarray*}
Let us find the value of $q$, which satisfies condition
({\ref{eq:dps:contraction_condit}}).
\begin{eqnarray*}
&& \underline{1}' \tilde{A} X  \leq \underline{1}'X,\\
&&  \sum_j x_j - (1-\rho) \sum_j
\frac{x_j}{1-\rho +f(\mu_j g_j)} \leq q \sum_j{x_j} \\
&& 1-(1-\rho) \frac{\sum_j \frac{x_j}{1-\rho +f(\mu_j
g_j)}}{\sum_j{x_j}} \leq q.
\end{eqnarray*}
As $f(\mu_j g_j)>0$ then
\begin{eqnarray*}
%&& 1-(1-\rho) \frac{\sum_j \frac{x_j}{1-\rho +f(\mu_j
%g_j)}}{\sum_j{x_j}} > 1-(1-\rho) \frac{\sum_j
%\frac{x_j}{1-\rho}}{\sum_j{x_j}} = 1-(1-\rho) \frac{1}{1-\rho}
%=0,\\
&& 0 < 1-(1-\rho) \frac{\sum_j \frac{x_j}{1-\rho +f(\mu_j
g_j)}}{\sum_j{x_j}} <1.
\end{eqnarray*}
Let us define $\delta$ in the following way:
\begin{eqnarray*}
&& \delta = \frac{1}{1-\rho+\max_j{f(\mu_j g_j)}} < \frac{\sum_j \frac{x_j}{1-\rho +f(\mu_j g_j)}}{\sum_j{x_j}}.
\end{eqnarray*}
Then
\begin{eqnarray*}
&& 1-(1-\rho) \frac{\sum_j \frac{x_j}{1-\rho +f(\mu_j g_j)}}{\sum_j{x_j}} < 1-(1-\rho)\delta.
\end{eqnarray*}
Let us notice that $\max_j{f(\mu_j g_j)}$ always exists as the
values of $\mu_j g_j$, $j=1,\ldots,M$ are finite. Then we can select
\begin{eqnarray*}
&& q=1-(1-\rho)\delta, \quad 0<q<1.
\end{eqnarray*}
Which completes the proof.
\end{proof}

\begin{Lemma}{\label{lemma:dps:y_n_to_y}}
If
\begin{eqnarray}
&& y_1^{(0)}=[0,\ldots,0],\label{eq:dps:y_n_zero}\\
&& y^{(n)}=\tilde{\mu}+y^{(n-1)} \tilde{A}, \quad n=1,2,\ldots ,
\label{eq:dps:y_n_th}
\end{eqnarray}
then $y^{(n)}\rightarrow y$, when $n \rightarrow \infty$.
\end{Lemma}
\begin{proof}
Let us present $y$ in the following way. As $B=E-A-D$, then
\begin{eqnarray*}
&& y=\underline{1}(E-B)^{-1}M,\\
&& y M^{-1}(E-D-A)=\underline{1},\\
&& y M^{-1}(E-D)= -y M^{-1}A+\underline{1},\\
&& y (E-D)^{-1}M= -y M^{-1}A(E-D)^{-1}M+\underline{1}(E-D)^{-1}M.
\end{eqnarray*}
As matrixes $D$ and $M$ are diagonal, the $MD=DM$ and then
\begin{eqnarray*}
&& y={\mu}^{T}(E-D)^{-1}+y M^{-1}A M (E-D)^{-1},
\end{eqnarray*}
where $\mu=[\mu_1,\ldots,\mu_M]$. According to notations
(\ref{eq:dps:mu_tilde_def}) and (\ref{eq:dps:A_tilde_def}) we have
the following
\begin{eqnarray*}
&& y=\tilde{\mu}+y \tilde{A}.
\end{eqnarray*}
Let us denote \mbox{$y^{(n)}=[y_1^{(n)},\ldots,y_1^{(n)}]$},
$n=0,1,2,\ldots$ and let define $y_1^{(0)}$ and
$y^{(n)}$ by (\ref{eq:dps:y_n_zero}) and ( \ref{eq:dps:y_n_th}). According to
Lemma~{\ref{lemma:dps:tilde_A_contraction}} reflection $\tilde{A}$
is a positive reflection and is a contraction. Also
$\tilde{\mu}_i$ are positive. Then $y^{(n)}\rightarrow y$, when $n
\rightarrow \infty$ and we prove the statement of Lemma~{\ref{lemma:dps:y_n_to_y}}.
\end{proof}

\begin{Lemma} {\label{lemma:dps:y_n_decreases}}
Let $y^{(n)}$ is defined by (\ref{eq:dps:y_n_th}) and
$y^{(0)}$  is given by (\ref{eq:dps:y_n_zero}), then
\begin{eqnarray} {\label{eq:dps:y_i_inequality}}
&& y_1^{(n)} \geq y_2^{(n)} \geq \ldots\geq y_M^{(n)}, \quad
n=1,2,\ldots
\end{eqnarray}
if $\frac{\mu_{i+1}}{\mu_{i}} \leq 1-\rho$ for every
$i=1,\ldots,M$.
\end{Lemma}
\begin{proof} We prove the statement ({\ref{eq:dps:y_i_inequality}})  by
induction. For $ y^{(0)}$ the statement
({\ref{eq:dps:y_i_inequality}}) is true. Let us assume that
({\ref{eq:dps:y_i_inequality}}) is true for the $(n-1)$ step,
\mbox{$y_1^{(n-1)} \geq y_2^{(n-1)} \geq \ldots\geq y_M^{(n-1)}$}. To
prove the induction statement we have to show that \mbox{$y_1^{(n)}
\geq y_2^{(n)} \geq \ldots\geq y_M^{(n)}$}, which is equal to that
$y^{(n)}_j \geq y^{(n)}_p$, if $j \leq p$. As
\begin{eqnarray*}
&& y^{(n)}_j=\tilde{\mu}_j+ \sum_{i=1}^M y^{(n-1)}_i
\tilde{A_{ij}},
\end{eqnarray*}
then
\begin{eqnarray*}
y^{(n)}_j - y^{(n)}_p  &=& \tilde{\mu}_j+ \sum_{i=1}^M y^{(n-1)}_i
\tilde{A_{ij}} - \left(\tilde{\mu}_p+ \sum_{i=1}^M y^{(n-1)}_i
\tilde{A_{ip}}\right)  = \\
& = &  \tilde{\mu}_j -  \tilde{\mu}_p + \sum_{i=1}^M y^{(n-1)}_i
(\tilde{A_{ij}} - \tilde{A_{ip}}).
\end{eqnarray*}
To show that $y^{(n)}_j - y^{(n)}_p$ we need to show that
$\tilde{\mu}_j - \tilde{\mu}_p \geq 0$ and $\sum_{i=1}^M
y^{(n-1)}_i (\tilde{A_{ij}} - \tilde{A_{ip}}) \geq 0 $, when $j
\leq p$. Let us show that to prove that $\sum_{i=1}^M y^{(n-1)}_i
(\tilde{A_{ij}} - \tilde{A_{ip}}) \geq 0 $, $j \leq p$ it is
enough to prove that $\sum_{i=1}^{r}(\tilde{A}_{ij} -
\tilde{A}_{ip}) \geq 0$, $j \leq p$, $r=1,\ldots,M$.  If we
regroup this sum we can get the following
\begin{eqnarray*}
\sum_{i=1}^M y^{(n-1)}_i (\tilde{A_{ij}} - \tilde{A_{ip}}) = && \\
\sum_{i=1}^M (y^{(n-1)}_i-y^{(n-1)}_{i+1}+y^{(n-1)}_{i+1}-\ldots-y^{(n-1)}_M+y^{(n-1)}_M) (\tilde{A_{ij}} - \tilde{A_{ip}}) = && \\
=\sum_{i=1}^{M-1} (y^{(n-1)}_i -
y^{(n-1)}_{i+1})\left[(\tilde{A}_{1j} -
\tilde{A}_{1p})+(\tilde{A}_{2 j} - \tilde{A}_{2 p})+\ldots+
(\tilde{A}_{i j} - \tilde{A}_{ip})\right] +&& \\+ y^{(n-1)}_M (
(\tilde{A}_{1j} - \tilde{A}_{1p})+ \ldots+ (\tilde{A}_{(M-1)j} -
\tilde{A}_{(M-1)p}) + (\tilde{A}_{Mj} - \tilde{A}_{Mp})) =  &&
\end{eqnarray*}
\begin{eqnarray*}
= \sum_{i=1}^{M-1}(y^{(n-1)}_i - y^{(n-1)}_{i+1})
\sum_{k=1}^r(\tilde{A}_{kj} - \tilde{A}_{kp})+ y^{(n-1)}_M
\sum_{k=1}^M (\tilde{A}_{kj} - \tilde{A}_{kp}).
\end{eqnarray*}
As $y^{(n-1)}_{i} \geq y^{(n-1)}_{i+1}$, $i=1,\ldots,M$, according
to the induction step, then to show that $\sum_{i=1}^M y^{(n-1)}_i
(\tilde{A_{ij}} - \tilde{A_{ip}}) \geq 0 $, $j \leq p$ it is
enough to show that $\sum_{i=1}^{r}(\tilde{A}_{ij} -
\tilde{A}_{ip}) \geq 0$, $j \leq p$, $r=1,\ldots,M$. We show this
in Lemma~{\ref{lemma:tilde_A_inequal}}. In
Lemma~{\ref{lemma:tilde_mu_inequal}} we show $\tilde{\mu}_j \geq
\tilde{\mu}_p$, $j\leq p$, when $\frac{\mu_{i+1}}{\mu_{i}} \leq
1-\rho$ for every $i=1,\ldots,M$. Then we prove the induction
statement and so prove the statement of
Lemma~{\ref{lemma:dps:y_n_decreases}}.
\end{proof}

\begin{Lemma}{\label{lemma:tilde_mu_inequal}}
\begin{eqnarray}
&& \tilde{\mu}_1 \geq \tilde{\mu}_2 \ldots \geq  \tilde{\mu}_M,
\label{eq:dps:mu_tilde_ineq}
\end{eqnarray}
if  $\frac{\mu_{i+1}}{\mu_{i}} \leq 1-\rho$ for every
$i=1,\ldots,M$.
\end{Lemma}
\begin{proof}
Let us compare $\tilde{\mu}_j$ and $\tilde{\mu}_p$, $j \leq p$. If
$\mu_j=\mu_p$ and $g_j=g_p$, then $\tilde{\mu}_j=\tilde{\mu}_p$
and (\ref{eq:dps:mu_tilde_ineq}) is satisfied.  Let us denote
\begin{eqnarray*}
&& f_2(x)= \sum_{k} \frac{g_k}{x+\mu_k g_k},
\end{eqnarray*}
which has the following properties
\begin{eqnarray}
&& 0 <f_2(x) <\rho. \label{eq:dps:f_2_prop}
\end{eqnarray}
Then
\begin{eqnarray*}
\tilde{\mu}_i &=& \frac{\mu_i}{1-\sum_j{\frac{g_j}{\mu_i
g_i+\mu_jg_j}}} = \frac{\mu_i}{1-f_2(\mu_i g_i)}.
\end{eqnarray*}
Let us find
\begin{eqnarray*}
\tilde{\mu}_j - \tilde{\mu}_p &=& \frac{\mu_j}{1-f_2(\mu_j g_j)}-
\frac{\mu_p}{1-f_2(\mu_p g_p)} =\frac{\mu_j-\mu_p - (\mu_j
f_2(\mu_p g_p) - \mu_p f_2(\mu_j g_j) )}{(1-f_2(\mu_j g_j))
(1-f_2(\mu_p g_p))}.
\end{eqnarray*}
 As (\ref{eq:dps:f_2_prop}) then
\begin{eqnarray*}
&& \mu_j f_2(\mu_p g_p) - \mu_p f_2(\mu_j g_j) < \mu_j \rho.
\end{eqnarray*}
Then
\begin{eqnarray*}
\tilde{\mu}_j - \tilde{\mu}_p & >&
\frac{(\mu_j-\mu_p)}{(1-f_2(\mu_j g_j) (1-f_2(\mu_p g_p)))} \left(
1- \rho\left(\frac{\mu_j}{\mu_j-\mu_p}\right) \right) =\\
& =& \frac{(\mu_j-\mu_p)}{(1-f_2(\mu_j g_j) (1-f_2(\mu_p g_p)))}
\left( 1- \rho\left(\frac{1}{1-\frac{\mu_p}{\mu_j}}\right) \right)
\geq 0,
\end{eqnarray*}
when
\begin{eqnarray*}
&& \frac{\mu_p}{\mu_{j}} \leq 1-\rho.
\end{eqnarray*}
So, if  $\frac{\mu_p}{\mu_j} \leq 1-\rho$ and $g_j \geq g_p$, then
$\tilde{\mu}_j \geq \tilde{\mu}_p$. Let us show that if $\mu_j >
\mu_p$ and $g_j=g_p$, then $\tilde{\mu_j} \geq \tilde{\mu_p}$. In
this case

\begin{eqnarray*}
\tilde{\mu}_j - \tilde{\mu}_p &=& \frac{\mu_j}{1-f_2(\mu_j g_j)}-
\frac{\mu_p}{1-f_2(\mu_p g_p)} = \\
& =& \frac{\mu_j-\mu_p - (\mu_j f_2(\mu_p g_j) - \mu_p f_2(\mu_j
g_j) )}{(1-f_2(\mu_j g_j)) (1-f_2(\mu_p g_j))}\\
&=& \frac{\Delta_1}{(1-f_2(\mu_j g_j)) (1-f_2(\mu_p g_j))}.
\end{eqnarray*}

Let us find when $\Delta_1>0$.

\begin{eqnarray*}
&& \Delta_1 = {\mu_j-\mu_p - \left(\mu_j \sum_{k=1}^M
\frac{g_k}{\mu_p
g_j+\mu_k g_k} - \mu_p \sum_{k=1}^M \frac{g_k}{\mu_j g_j+\mu_k g_k}\right)} =\\
&& ={\mu_j-\mu_p - \left(\sum_{k=1}^M
\frac{g_k(g_j(\mu_j^2-\mu_p^2)+\mu_k g_k(\mu_j-\mu_p))}{(\mu_p
g_j+\mu_k g_k)(\mu_p g_j+\mu_k g_k)}\right)} \\
&& = (\mu_j-\mu_p) \left(1-\sum_{k=1}^M
\frac{g_k(g_j(\mu_j+\mu_p)+\mu_k g_k)}{(\mu_p g_j+\mu_k g_k)(\mu_p
g_j+\mu_k g_k)}\right).
\end{eqnarray*}
As
\begin{eqnarray*}
&& 0 < \mu_j \mu_p g_j^2, \quad k=1,\ldots,M , \\
&& {g_k \mu_k (g_j(\mu_j+\mu_p)+\mu_k g_k)}< {(\mu_j \mu_p g_j^2 +
g_j(\mu_j+\mu_p)\mu_k g_k+\mu_k^2 g_k^2)}, \quad k=1,\ldots,M \\
&& {g_k \mu_k (g_j(\mu_j+\mu_p)+\mu_k g_k)}< {(\mu_j g_j+\mu_k
g_k)(\mu_p g_j+\mu_k g_k)}, \quad k=1,\ldots,M \\
&& \frac{g_k(g_j(\mu_j+\mu_p)+\mu_k g_k)}{(\mu_j g_j+\mu_k
g_k)(\mu_p g_j+\mu_k g_k)} < \frac{1}{\mu_k}, \quad k=1, \ldots,M.
\end{eqnarray*}
Then
\begin{eqnarray*}
&&\Delta_1
 > (\mu_j-\mu_p) \left(1-\sum_{k=1}^M{\frac{1}{\mu_k}}\right)=1-\rho > 0.
\end{eqnarray*}

Then we proved the following:
\begin{eqnarray*}
&& \text{If} \quad \mu_j=\mu_p, \,\, g_j=g_p, \quad \text{ then}
\quad \tilde{\mu_j}=\tilde{\mu_p}, \\
&& \text{If}\quad \mu_j>\mu_p, \,\, g_j=g_p, \quad\text{
then}\quad \tilde{\mu_j}>\tilde{\mu_p}, \\
&& \text{If}\quad \frac{\mu_p}{\mu_j} \leq 1-\rho,
\,\, g_j \geq g_p, \quad\text{ then}\quad \tilde{\mu_j}\geq
\tilde{\mu_p}.
\end{eqnarray*}
Setting $p=j+1$ and remembering that $\mu_1 \geq \ldots \geq
\mu_M$, we get that $\tilde{\mu}_1 \geq \tilde{\mu}_2\ldots\geq
\tilde{\mu}_M$ is true when $\frac{\mu_{i+1}}{\mu_i} \leq 1-\rho$
for every $i=1,\ldots,M$. That proves the statement of
Lemma~{\ref{lemma:tilde_mu_inequal}}.

Returning back to the main Theorem~{\ref{theorem:dps:T_monot}},
Lemma~{\ref{lemma:tilde_mu_inequal}} gives condition
({\ref{eq:dps:condit_satisf}}) as a restriction on a system
parameters.

Let us notice that for the job classes such for which the means
are such as $\frac{\mu_{i+1}}{\mu_i} < 1-\rho$, if the weights
given for these classes are equal, then still $\tilde{\mu}_i \geq
\tilde{\mu}_{i+1}$. This condition gives us as a result
Remark~\ref{remark:dps:weight_sel_for_y} and
Remark~{\ref{remark:dps:weight_sel}}.
\end{proof}

\begin{Lemma}{\label{lemma:tilde_A_inequal}}
\begin{eqnarray*}
&& \sum_{i=1}^r {\tilde{A}_{i 1}} \geq  \sum_{i=1}^r {\tilde{A}_{i
2}} \geq \ldots \geq \sum_{i=1}^r {\tilde{A}_{i M}}, \quad r=1,\ldots,M.
\end{eqnarray*}
\end{Lemma}
\begin{proof}
Let us remember $\tilde{A}=M^{-1}AM(E-D)^{-1}$. Then as
$\rho=\sum_{k=1}^M{\frac{1}{\mu_k}}$, then
\begin{eqnarray*}
&& \sum_{i=1}^r {\tilde{A}_{ij}} = \frac{\sum_{i=1}^r{\frac{ \mu_j
g_j}{\mu_i(\mu_j g_j+\mu_i g_i)}}}{1-\sum_{k=1}^M
{\frac{g_k}{\mu_j g_j+ \mu_k g_k}}} =\frac{\sum_{i=1}^r{\frac{
\mu_j g_j}{\mu_i(\mu_j g_j+\mu_i g_i)}}}{1-\rho+\sum_{k=1}^M
{\frac{\mu_j g_j}{\mu_k(\mu_j g_j+ \mu_k g_k)}}}
\end{eqnarray*}
Let us define
\begin{eqnarray*}
&&  f_3(x) = \frac{\sum_{i=1}^r {\frac{x}{\mu_i(x+\mu_i
g_i)}}}{1-\rho + \sum_{k=1}^M{\frac{x}{\mu_k(x+\mu_k g_k)}}} =
\frac{h_1(x)}{1-\rho + h_1(x)+h_2(x)},
\end{eqnarray*}
where
\begin{eqnarray*}
&&  h_1(x) = \sum_{i=1}^r {\frac{x}{\mu_i(x+\mu_i g_i)}}>0, \\
&&  h_2(x) = \sum_{j=r+1}^M {\frac{x}{\mu_j(x+\mu_j g_j)}}>0.
\end{eqnarray*}
Let us show that $f_3(x)$ is increasing on $x$. For that it enough
to show that $\frac{d f_3(x)}{dx} \geq 0$. Let us consider
\begin{eqnarray*}
\frac{d f_3(x)}{dx}& =& \frac{h_1'(x)(1-\rho) +
h_1'(x)h_2(x)-h_1(x)h_2'(x)}{(1-\rho + h_1(x)+h_2(x))^2}
\end{eqnarray*}
Since $h_1'(x)>0$ and $1-\rho>0$:
\begin{eqnarray*}
&& \frac{d f_3(x)}{dx}\geq0 \quad \mbox{if} \quad h_1'(x) h_2(x)-
h_1(x) h_2'(x)\geq0.
\end{eqnarray*}
Let us consider
\begin{eqnarray*}
h_1'(x) h_2(x)- h_1(x) h_2'(x) & = & \sum_{i=1}^r
{\frac{g_i}{(x+\mu_i g_i)^2}}\sum_{k=r+1}^M
{\frac{x}{\mu_k(x+\mu_k g_k)}} - \sum_{i=1}^r
{\frac{x}{\mu_i(x+\mu_i g_i)}}\sum_{k=r+1}^M
{\frac{g_k}{(x+\mu_k g_k)^2}} = \\
&=& \sum_{i=1}^r \sum_{k=r+1}^M \left({\frac{g_i x} {(x+ \mu_i
g_i)^2 (x+\mu_k g_k)\mu_k}}-{\frac{g_k x}
{\mu_i (x+\mu_i g_i)(x+\mu_k g_k)^2}}\right) = \\
&=& \sum_{i=1}^r \sum_{k=r+1}^M \frac{x}{(x+\mu_i g_i) (x+\mu_k
g_k)} \left({\frac{g_i}{\mu_k(x+\mu_i g_i)}}
-{\frac{g_k}{\mu_i(x+\mu_k g_k)}}\right) = \\
&=& \sum_{i=1}^r \sum_{k=r+1}^M \frac{x}{(x+\mu_i g_i) (x+\mu_k
g_k)} \left({\frac{\mu_i g_i (x+g_k \mu_k) -
\mu_k g_k (x+\mu_i g_i)} {\mu_i \mu_k (x+\mu_k g_k)(x+\mu_i g_i)}}\right) = \\
&=& \sum_{i=1}^r \sum_{k=r+1}^M \frac{x^2 \left(\mu_i g_i - \mu_k
g_k \right) }{(x+\mu_i g_i)^2(x+\mu_k g_k)^2 \mu_k\mu_i} \geq 0,
\end{eqnarray*}
Then $\frac{d f_3(x)}{dx} \geq 0$ and $f_3(x)$ is an increasing
function of $x$. As $\mu_j g_j \geq \mu_p g_p $, $j\leq p$, then
we prove the statement of Lemma~{\ref{lemma:tilde_A_inequal}}.
\end{proof}

\newpage
\tableofcontents


\begin{thebibliography}{1}


\bibitem{AAA06}
E. Altman, K. Avrachenkov, U. Ayesta, ``A survey on discriminatory
processor sharing'', Queueing Systems, Volume 53, Numbers 1-2,
June 2006 , pp. 53-63(11).

\bibitem{AJK04}
E. Altman, T. Jimenez, and D. Kofman, ``DPS queues with stationary
ergodic service times and the performance of TCP in overload,'' in
Proceedings of IEEE Infocom, Hong-Kong, March 2004.


\bibitem{AABN}
K. Avrachenkov, U. Ayesta, P. Brown, R. Nunez-Queija,
``Discriminatory Processor Sharing Revisited'',
INFOCOM 2005. 24th Annual Joint Conference of the
IEEE Computer and Communications Societies. Proceedings IEEE, Vol. 2 (2005), pp. 784-795 vol. 2.


\bibitem{BT01}
T. Bu and D. Towsley, ``Fixed point approximation for TCP behaviour in an AQM network'',
in Proceedings of ACM SIGMETRICS/Performance, pp. 216-225, 2001.

\bibitem{CBBLR05}
S. K. Cheung, J. L. van den Berg, R.J. Boucherie, R. Litjens, and F. Roijers,
``An analytical packet/flow-level modelling approach for
wireless LANs with quality-of-service support'', in Proceedings of ITC-19, 2005.

\bibitem{FMI90}
G. Fayolle, I. Mitrani, R. Iasnogorodski, ``Sharing a Processor
Among Many Job Classes'', Journal of the ACM (JACM),Volume 27,
Issue 3 (July 1980), pp. 519 - 532, 1980, ISSN:0004-5411.

\bibitem{GM02}
L. Guo and I. Matta, ``Scheduling flows with unknown sizes:
approximate analysis'', in Proceedings of the 2002 ACM SIGMETRICS
international conference on Measurement and modeling of computer
systems, 2002, pp. 276-277.


\bibitem{KNB05}
G. van Kessel, R. Nunez-Queija, S. Borst, ``Differentiated
bandwidth sharing with disparate flow sizes'', INFOCOM 2005. 24th
Annual Joint Conference of the IEEE Computer and Communications
Societies. Proceedings IEEE, Vol. 4 (2005), pp. 2425-2435 vol. 4.

\bibitem{KRNB04}
G. van Kessel, R. Nunez-Queija, S. Borst, ``Asymptotic regimes and
approximations for discriminatory processor sharing'', SIGMETRICS
Perform. Eval. Rev., vol. 32, pp. 44-46, 2004.


\bibitem{KK06}
B. Kim and J. Kim, ``Comparison of DPS and PS systems according to
DPS weights'', Communications Letters, IEEE, vol. 10, issue 7, pp.
558-560, ISSN:1089-7798, July 2006.

\bibitem{Kle67}
L. Kleinrock, ``Time-shared systems: A theoretical treatment'',
J.ACM, vol.14, no. 2, pp. 242-261, 1967.


\bibitem{HT05}
Y. Hayel and B. Tuffin, ``Pricing for heterogeneous services at a discriminatory
processor sharing queue'', in Proceedings of Networking, 2005.

\bibitem{RS94}
K. M. Rege, B. Sengupta, ``A decomposition theorem and related
results for the discriminatory processor sharing queue'', Queueing
Systems, Springer Netherlands ISSN 0257-0130, 1572-9443, Issue
Volume 18, Numbers 3-4 / September, 1994, pp. 333-351, 2005.


\bibitem{RS96}
K. M. Rege, B. Sengupta, ``Queue-Length Distribution for the
Discriminatory Processor-Sharing Queue'', Operations Research,
Vol. 44, No. 4 (Jul. - Aug., 1996), pp. 653-657.


\bibitem{RR94}
Rhonda Righter, "Scheduling", in M. Shaked and J. Shanthikumar  (eds),
"Stochastic Orders", 1994. 

\end{thebibliography}
\end{document}